\documentclass[preprint]{isipta2025}

\usepackage{enumitem}

\usepackage{nicefrac}
\newtheorem{runexample}{Running example}

\providecommand\given{}

\newcommand\SetSymbol{\nonscript\colon\allowbreak\nonscript\mathopen{}}
\DeclarePairedDelimiter{\group}{(}{)}

\DeclarePairedDelimiter\bra{\langle}{\rvert}
\DeclarePairedDelimiter\ket{\lvert}{\rangle}
\DeclarePairedDelimiterX\braket[2]{\langle}{\rangle}{#1\delimsize\vert#2}
\DeclarePairedDelimiterXPP{\btrace}[2]{\tr_{#1}}{(}{)}{}{{#2}}
\DeclarePairedDelimiterXPP{\trace}[1]{\tr}{(}{)}{}{{#1}}
\DeclarePairedDelimiterXPP{\inprod}[2]{\inne}{(}{)}{}{{#1,#2}}
\DeclarePairedDelimiterX{\set}[1]{\{}{\}}{\renewcommand\given{\SetSymbol}#1}
\DeclareMathOperator{\spec}{spec}

\DeclareMathOperator{\prj}{pr}

\DeclareMathOperator{\exte}{ext}
\DeclareMathOperator{\redu}{red}
\newcommand{\indifset}[1][]{\smash{\mathscr{I}_{#1}}}
\newcommand{\projindifset}[1][]{\indifset[{\proj[#1]}]}

\newcommand{\hilbertspace}[1][]{\mathscr{X}^{#1}}
\newcommand{\statespace}[1][]{\bar{\mathscr{X}}^{#1}}
\newcommand{\subspace}{{\mathcal{V}}}
\newcommand{\subunion}{\mathcal{S}}

\newcommand{\reals}{\mathbb{R}}

\newcommand{\nonnegreals}{\reals_{\geq0}}
\newcommand{\posreals}{\reals_{>0}}

\DeclareMathOperator{\spa}{span}
\newcommand{\linspanof}[2][]{\spa\group[#1]{#2}}

\newcommand{\densities}{\mathscr{R}}
\newcommand{\weakleq}{\leq}

\newcommand{\weakgeq}{\geq}
\newcommand{\weakgt}{\gneq}

\newcommand{\stronggt}{>}
\newcommand{\rngweakgeq}[1][]{\mathrel{\smash{\weakgeq_{\proj[#1]}}}}
\newcommand{\rngweakleq}[1][]{\mathrel{\smash{\weakleq_{\proj[#1]}}}}
\newcommand{\rngstronggt}[1][]{\mathrel{\smash{\stronggt_{\proj[#1]}}}}

\newcommand{\posdefmeasurements}{\smash{\measurements_{\stronggt\zero}}}
\newcommand{\rngposdefmeasurements}[1][]{\rng(\proj[#1])_{\stronggt\zero}}
\newcommand{\possemidefmeasurements}{\smash{\measurements_{\weakgeq\zero}}}
\newcommand{\rngpossemidefmeasurements}[1][]{\rng(\proj[#1])_{\weakgeq\zero}}
\newcommand{\newbackground}[1][\subspace]{\measurements_{\smash{\stronggt}}^{\smash{#1}}}
\newcommand{\projnewbackground}[1][]{\measurements_{\smash{\stronggt}}^{\smash{\proj[#1]}}}
\newcommand{\negdefmeasurements}{\smash{\measurements_{\stronggt\zero}}}
\newcommand{\negsemidefmeasurements}{\smash{\measurements_{\weakleq\zero}}}
\newcommand{\desirset}[1][]{\mathscr{D}_{#1}}
\newcommand{\assessment}{\mathscr{A}}

\newcommand{\utilities}{\mathscr{U}}
\newcommand{\credalset}{\mathscr{M}}
\DeclareMathOperator{\posi}{posi}
\DeclareMathOperator{\rng}{rng}
\newcommand{\natext}{\mathscr{E}}

\newcommand{\projlocalcond}[2][]{{#2}\rfloor\rng(\proj[#1])}
\newcommand{\localcond}[2][\subspace]{{#2}\rfloor{#1}}

\newcommand{\projcond}[2][]{{#2}\vert\rng(\proj[#1])}
\newcommand{\cond}[2][\subspace]{{#2}\vert{#1}}

\newcommand{\fket}[1][]{\ket{\phi_{#1}}}
\newcommand{\fbra}[1][]{\bra{\phi_{#1}}}

\newcommand{\gket}[1][]{\ket{\psi_{#1}}}
\newcommand{\gbra}[1][]{\bra{\psi_{#1}}}

\newcommand{\uket}{\ket{\Psi}}
\newcommand{\ubra}{\bra{\Psi}}

\DeclareMathOperator{\inne}{in}

\newcommand{\measurement}[1]{\hat{#1}}

\newcommand{\measurements}{\mathscr{H}}
\newcommand{\projection}[1][]{\measurement{P}_{#1}}
\newcommand{\density}[1][]{\measurement{\rho}_{#1}}
\DeclareMathOperator{\tr}{Tr}
\newcommand{\identity}{\measurement{I}}
\newcommand{\zero}{\measurement{0}}
\newcommand{\orth}{\perp}
\newcommand{\val}{u}

\newcommand{\utility}[1]{\val_{\measurement{#1}}}

\newcommand{\proj}[1][]{\prj_{#1}}
\newcommand{\invproj}[1][]{\smash{\prj^{-1}_{#1}}}
\newcommand{\projof}[2][]{\proj[#1]\group{\measurement{#2}}}
\newcommand{\reduce}[1][]{\smash{\redu_{\proj[#1]}}}
\newcommand{\invreduce}[1][]{\smash{\redu^{-1}_{\proj[#1]}}}
\newcommand{\reduceof}[2][]{\smash{\redu_{\proj[#1]}\group{\measurement{#2}}}}
\newcommand{\extend}[1][]{\smash{\exte_{\proj[#1]}}}
\newcommand{\invextend}[1][]{\smash{\exte^{-1}_{\proj[#1]}}}
\newcommand{\extendof}[2][]{\smash{\exte_{\proj[#1]}\group{\measurement{#2}}}}

\newcommand{\eigval}{\lambda}
\newcommand{\eigket}[1][]{\ket{a_{#1}}}
\newcommand{\eigbra}[1][]{\bra{a_{#1}}}

\newcommand{\eigspace}{{\mathscr{E}}}

\DeclareMathOperator{\expe}{E}
\newcommand{\expectation}[1][]{\expe_{#1}}

\newcommand{\prevsymbol}{\mathrm{P}}
\newcommand{\linprev}[1][]{\smash{\prevsymbol_{#1}}}
\newcommand{\lowprev}[1][]{\smash{\underline{\prevsymbol}}_{#1}}
\newcommand{\uppprev}[1][]{\smash{\overline{\prevsymbol}}_{#1}}

\newcommand{\lprice}[1][]{\smash{\underline{\Lambda}}_{#1}}
\newcommand{\uprice}[1][]{\smash{\overline{\Lambda}}_{#1}}

\newcommand{\con}{\alpha}

\newcommand{\then}{\Rightarrow}
\newcommand{\ifandonlyif}{\Leftrightarrow}
\newcommand{\spectrum}[1]{\spec(\measurement{#1})}
\newcommand{\spectrumpure}[1]{\spec(#1)}
\newcommand{\utilitypure}[1]{\val_{#1}}
\newcommand{\bolleke}{\vcenter{\hbox{\scalebox{1}{\(\bullet\)}}}}

\title{
  Conditioning through indifference in quantum mechanics
}
\author[1]{Keano De Vos}
\author[1]{Gert de Cooman}
\affil[1]{Foundations Lab for imprecise probabilities, Ghent University, Belgium}

\begin{document}
\maketitle

\begin{abstract}
We can learn (more) about the state that a quantum system is in through measurements.
We look at how to describe the uncertainty about a quantum system's state conditional on executing such measurements.
We show that by exploiting the interplay between desirability, coherence and indifference, a general rule for conditioning can be derived.
We then apply this rule to conditioning on measurement outcomes, and show how it generalises to conditioning on a set of measurement outcomes.
\end{abstract}


\section{Introduction}\label{sec:introduction}
Since we can learn (more) about the state that a quantum system is in through measurements, it's of paramount importance to understand how to make decisions based on the outcomes of these measurements.
This is the central question we'll answer here: how to represent the uncertainty about a quantum system's state conditional on the outcome of a measurement?

This question has been addressed in the literature, most famously by \citeauthor{luders1950} \cite{luders1950}, but most of the extant answers are based on the --- arguably too narrow --- notion of probability in quantum mechanics.
That is why we build upon the more general sets of desirable measurements framework, which goes back to \citeauthor{Benavoli_2016} \cite{Benavoli_2016}, and which we recently explored and tried to justify in \cite{devos2025:quantum:decision}.
In that earlier work of ours on desirable measurements, we developed a decision-theoretic argument involving imprecise probabilities to model the uncertainty about a quantum system's state.
This has led to a similar mathematical framework as that first introduced by \citeauthor{Benavoli_2016} \cite{Benavoli_2016}, but with a different interpretation.

Our argument there proceeds along the following lines.
The system is in an unknown state~\(\uket\) in the state space~\(\statespace\).
Hermitian operators \(\measurement{A}\) represent measurements, which allow us to interact with the system and learn more about the system state.
We associate with each such measurement operator a utility function~\(\utility{A}\colon\statespace\to\reals\), which represents the reward associated with performing that measurement.
If the system is in state~\(\gket\), then \(\utility{A}(\gket)\) is the utility associated with performing the measurement~\(\measurement{A}\): performing measurement~\(\measurement{A}\) on the system in the unknown state~\(\uket\) results in an uncertain reward~\(\utility{A}(\uket)\).
Through decision-theoretic postulates based on the non-probabilistic foundations of quantum mechanics, we show that this utility function must have the form~\(\utility{A}(\fket)=\fbra\measurement{A}\fket\) for all~\(\fket\in\statespace\).\footnote{Our argumentation in \cite{devos2025:quantum:decision}, though different from Deutsch's in \cite{deutsch1988}, therefore comes to the same conclusions.}
The uncertainty of a rational subject --- called You --- about the unknown state~\(\uket\) can then be captured by expressing preferences between measurements, via preferences between their associated uncertain rewards.
Such a (partial) preference ordering on measurements is therefore a model for Your uncertainty about~\(\uket\).
Equivalently, You can use a so-called \emph{set of desirable measurements}, which are those measurements You prefer to the status quo --- the null measurement.
Modelling uncertainty in this way follows the sets of desirable gambles approach that's common in imprecise probabilities research; see also \cite{ITIP,Couso,Lesot_coherent_2020,miranda2010,moral,quaeghebeur2015:statement,lower_previsions,Walley}.

Here, using the power of indifference statements, already proven invaluable in our earlier work \cite{devos2023:indistinguishability}, we'll show how to extend this framework to deal with conditioning in quantum mechanics.
After a concise introduction to the desirable measurements approach in \cref{sec:desirability}, we have a closer look in \cref{sec:conditioning} at how to represent Your new knowledge of a measurement outcome in this framework.
In \cref{sec:abstract}, we look at an abstract representation of updating through indifference statements, which we then apply in \cref{sec:back:to:conditioning} to conditioning on measurement outcomes.
We thus retrieve a conditioning rule similar to the one by \citeauthor{Benavoli_2016} \cite{Benavoli_2016} but with a different interpretation and broader scope.
Finally, in \cref{sec:general:conditioning}, we show how our abstract approach allows us to extend this to conditioning on a set of measurement outcomes, as a first step towards conditioning on positive operator valued measures.

\section{Desirability in quantum mechanics}\label{sec:desirability}
To make this paper self-contained, we'll first provide a brief overview of the framework we're using for dealing with uncertainty in quantum mechanics, and revisit some of the more relevant and important concepts in the sets of desirable measurements framework first introduced by \citeauthor{Benavoli_2016} \cite{Benavoli_2016}, which we later provided a decision-theoretic justification for; see for instance \cite{devos2025:quantum:decision}.

\subsection{Quantum mechanics}
The framework is based on combining ideas from decision theory with the non-probabilistic principles of quantum mechanics; for an account of the basics of quantum mechanics, see \citep{Cohen1,Nielsen}.

The \emph{state}~\(\gket\) of a quantum system is a normalised element of a complex \emph{Hilbert space}~\(\hilbertspace\).
To keep the discussion simple, \emph{we'll restrict ourselves to finite-dimensional Hilbert spaces}, which can for instance be used to model the spin or (with some extra assumptions) the energy of a bounded electron.
Such finite-dimensional spaces are particularly useful in quantum computing and quantum cryptography \cite{Nielsen}.
We'll use the Dirac notation: a \emph{ket}~\(\gket\) is a vector in \(\hilbertspace\), and the \emph{bra}~\(\gbra\) its adjoint.
The \emph{state space}~\(\statespace\) contains all normalised kets.

A \emph{measurement} on the system is represented by a Hermitian operator~\(\measurement{A}\coloneqq\sum_{k=1}^n\eigval_k\eigket[k]\eigbra[k]\) on~\(\hilbertspace\), where the real numbers \(\eigval_1\), \dots, \(\eigval_n\) are its eigenvalues and the kets \(\eigket[1]\), \dots, \(\eigket[n]\) in \(\statespace\) its corresponding pairwise orthonormal eigenkets.
The possible outcomes of a measurement \(\measurement{A}\) are the eigenvalues~\(\eigval_1,\dots,\eigval_n\), which we collect in its \emph{spectrum} \(\spectrum{A}\coloneqq\set{\eigval_1,\dots,\eigval_n}\).
We denote the real linear space of all such Hermitian operators by~\(\measurements(\hilbertspace)\), or simply \(\measurements\) if no confusion is possible.

\subsection{Utility functions}
We want to represent beliefs of a subject, whom we'll call You, about the unknown quantum mechanical state~\(\uket\) of a system.
You can interact with the system through measurements~\(\measurement{A}\in\measurements\), which we can see as possible acts or options.
As is common in decision theory \cite{anscombe1963,aumann1962,aumann1964,Lesot_coherent_2020,Definetti,nau2006,Walley,zaffalon2017:incomplete:preferences}, Your uncertainty will be described by Your preferences between these different acts, and we attach to each such act/measurement~\(\measurement{A}\) a utility function~\(\utility{A}\colon\statespace\to\reals\), where \(\utility{A}(\fket)\) is the reward associated with performing the measurement~\(\measurement{A}\) when the system is in state~\(\fket\), expressed in units of some linear utility.

We have argued elsewhere \cite{devos2025:quantum:decision} that a number of decision-theoretic principles, based on the non-probabilistic postulates of quantum mechanics, leave You with no choice about which utility functions~\(\utility{A}\) to use: they require them to take the form \(\utility{A}(\fket)=\fbra\measurement{A}\fket\) for all~\(\fket\in\statespace\).
The linear space~\(\utilities\) of all utility functions is therefore linearly isomorphic to the real linear space~\(\measurements\): we can --- and will --- identify measurements and their utility functions.

\subsection{Desirability}
The idea behind our approach is that Your uncertainty about the system's unknown state~\(\uket\) can be modelled through a strict partial preference ordering between uncertain rewards~\(\utility{A}(\uket)\), which is equivalent to a strict partial vector ordering on the linear space~\(\measurements\).
A mathematically equivalent model for such a preference relation is a set of desirable utility functions: those uncertain rewards that You strictly prefer to the zero utility function~\(0\), or equivalently, a \emph{set of desirable measurements}~\(\desirset\subseteq\measurements\): those measurements that You strictly prefer to the status quo~\(\zero\).\footnote{The operator \(\zero\) is the unique Hermitian operator all of whose eigenvalues are zero, and whose utility function \(\utility{0}\) is identically zero.}

Typically, rationality criteria are then imposed on such a set of desirable utility functions \cite{ITIP,Lesot_coherent_2020,Walley}, which can be readily translated to the desirable measurements framework.
We call a set of desirable measurements~\(\desirset\) \emph{coherent} if for all~\(\measurement{A},\measurement{B}\in\measurements\) and all~\(\lambda\in\posreals\):\footnote{These are the rationality criteria for Accept-Favour models \cite{quaeghebeur2015:statement}, based on a set of favourable options \(\desirset\), that respect the background model consisting of the accepted measurements \(\possemidefmeasurements\) and favourable measurements \(\negdefmeasurements\).}
\begin{enumerate}[label={\upshape D\arabic*.},ref={\upshape D\arabic*},labelwidth=*,leftmargin=*,itemsep=0pt]
\item\label{axiom:desirs:consistency} \(\negsemidefmeasurements\cap\desirset=\emptyset\);\hfill[rejecting possible loss]
\item\label{axiom:desirs:background} \(\posdefmeasurements\subseteq\desirset\);\hfill[accepting sure gain]
\item\label{axiom:desirs:monotonicity} if \(\measurement{A}\in\desirset\) and \(\measurement{B}\weakgeq\measurement{A}\) then \(\measurement{B}\in\desirset\);\hfill[monotonicity]
\item\label{axiom:desirs:additivity} \(\measurement{A},\measurement{B}\in\desirset\then\measurement{A}+\measurement{B}\in\desirset\);\hfill[additivity]
\item\label{axiom:desirs:scaling} \(\measurement{A}\in\desirset\then\lambda\measurement{A}\in\desirset\).\hfill[positive scaling]
\end{enumerate}
Here, \(\stronggt\), \(\weakgeq\) and \(\weakgt\) are the vector orderings defined by
\begin{multline*}
\begin{aligned}
\measurement{A}\stronggt\measurement{B}&\ifandonlyif\min\spectrumpure{\measurement{A}-\measurement{B}}>0\\
\measurement{A}\weakgeq\measurement{B}&\ifandonlyif\min\spectrumpure{\measurement{A}-\measurement{B}}\geq0\\
\measurement{A}\weakgt\measurement{B}&\ifandonlyif\min\spectrumpure{\measurement{A}-\measurement{B}}\geq0\text{ and }\measurement{A}\neq\measurement{B},
\end{aligned}\\
\text{ for all~\(\measurement{A},\measurement{B}\in\measurements\),}
\end{multline*}
\(\posdefmeasurements\coloneqq\set{\measurement{A}\in\measurements\given\measurement{A}\stronggt\zero}\) is the set of \emph{positive definite} measurements and \(\possemidefmeasurements\coloneqq\set{\measurement{A}\in\measurements\given\measurement{A}\weakgeq\zero}\) the set of \emph{positive semidefinite} measurements.
In \cite{Benavoli_2016}, a similar framework involving sets of desirable measurements was used, with a different interpretation and justification, and with a slightly stronger version of~\ref{axiom:desirs:background}.\footnote{Indeed, often in similar contexts, a stronger requirement, such as accepting partial gains, is imposed. We prefer the weaker requirement here in the context of our discussion of indifference further on.}

One interesting aspect of working with partial preference models in the form of coherent sets of desirable measurements, is that they allow for conservative inference; see \cite[Section 3.7]{Walley} and also \cite{quaeghebeur2015:statement,cooman2005:melief:models}.
Suppose You identify some assessment of measurements \(\assessment\subseteq\measurements\) that You deem desirable, then the rationality criteria allow us to find the smallest --- if any --- coherent set of desirable measurements~\(\desirset\) that includes~\(\assessment\).
The constructive criteria \labelcref{axiom:desirs:background,axiom:desirs:monotonicity,axiom:desirs:additivity,axiom:desirs:scaling} prompt us to consider the set
\begin{align}
\natext(\assessment)
\coloneqq&\posi\group[\big]{(\assessment+\possemidefmeasurements)\cup\posdefmeasurements}
\notag\\
=&\posdefmeasurements\cup\group{\posi\group{\assessment}+\possemidefmeasurements},
\label{eq:natex}
\end{align}
where \(\posi\) denotes the positive linear hull.
If \(\natext(\assessment)\) satisfies the destructive criterion \labelcref{axiom:desirs:consistency}, or equivalently, if \(\posi(\assessment)\cap\negsemidefmeasurements=\emptyset\), then we'll call the assessment \(\assessment\) \emph{consistent}.
It can then --- and only then --- be extended to a coherent set of desirable measurements, and the \emph{natural extension} \(\natext(\assessment)\) is then the smallest such set.

\subsection{Coherent (lower and upper) previsions}
With a set of desirable measurements, we can associate a lower price functional~\(\lprice[\desirset]\) and an upper price functional~\(\uprice[\desirset]\) as follows: for all \(\measurement{A}\in\measurements,\)\footnote{\(\identity\) is the identity operator, defined by \(\identity\gket\coloneqq\gket\) for all \(\gket\in\hilbertspace\).}
\begin{gather}
\lprice[\desirset](\measurement{A})
\coloneqq\sup\set{\con\in\reals\given\measurement{A}-\con\identity\in\desirset}
\label{eq:lowprevv}\\
\uprice[\desirset](\measurement{A})
\coloneqq\inf\set{\con\in\reals\given\con\identity-\measurement{A}\in\desirset}.
\label{eq:upprevv}
\end{gather}
The lower price~\(\lprice[\desirset](\measurement{A})\) is Your supremum buying price for the measurement \(\measurement{A}\), or equivalently, for the uncertain reward \(\utility{A}(\uket)\).
The upper price~\(\uprice[\desirset](\measurement{A})\) is Your infimum selling price for the measurement~\(\measurement{A}\).
Observe that \(\lprice[\desirset](\measurement{A})=-\uprice[\desirset](-\measurement{A})\).
It's well-known that the lower price functional~\(\lprice[\desirset]\) fully characterises the coherent set~\(\desirset\) up to border behaviour; see for instance \cite[Sec.~3.8]{Walley} and \cite{Lesot_coherent_2020}.
In this sense, lower price functionals and sets of desirable measurements are (almost) equivalent mathematical models for Your beliefs.

A real functional~\(\lowprev\) on~\(\measurements\) is called a \emph{coherent lower prevision} if there's some coherent set of desirable measurements~\(\desirset\) such that \(\lowprev=\lprice[\desirset]\).
The coherence of a lower prevision is characterised by the following properties: for any~\(\measurement{A},\measurement{B}\in\measurements\) and~\(\lambda\in\nonnegreals\),
\begin{enumerate}[label={\upshape LP\arabic*.},ref={\upshape LP\arabic*},labelwidth=*,leftmargin=*,itemsep=0pt,series=LP,widest=3]
\item\label{it:lowprev:superadditivity} \(\lowprev(\measurement{A}+\measurement{B})\geq\lowprev(\measurement{A})+\lowprev(\measurement{B})\);\hfill\textup{[super-additivity]}
\item\label{it:lowprev:homogeneity} \(\lowprev(\lambda\measurement{A})=\lambda\lowprev(\measurement{A})\);\hfill\textup{[non-negative homogeneity]}
\item\label{it:lowprev:bounds} \(\lowprev(\measurement{A})\geq\min\spectrum{A}\);\footnote{This is equivalent to C3 in \cite{ITIP}, as \(\min\utility{A}=\min\spectrum{A}\) \protect{\cite{devos2025:quantum:decision}}.}\hfill\textup{[accepting sure gains]}
\end{enumerate}
see \cite[C1--C3 and Prop.~2.2]{ITIP} and \cite{Lesot_coherent_2020,Walley}.
We'll denote the \emph{conjugate} upper prevision by \(\uppprev\), where \(\uppprev(\bolleke)\coloneqq-\lowprev(-\,\bolleke)\).
When a coherent lower prevision is self-conjugate, so if \(\lowprev=\uppprev\), we call it a \emph{linear prevision}, or a \emph{coherent prevision}, and simply denote it as~\(\linprev\).
Clearly, \(\linprev\) is a coherent prevision if and only if for all~\(\measurement{A},\measurement{B}\in\measurements\) and all~\(\lambda,\mu\in\reals\),
\begin{enumerate}[label={\upshape P\arabic*.},ref={\upshape P\arabic*},labelwidth=*,leftmargin=*,itemsep=0pt,series=P,widest=2]
\item\label{axiom:prev:linearity} \(\linprev(\lambda\measurement{A}+\mu\measurement{B})=\lambda\linprev(\measurement{A})+\mu\linprev(\measurement{B})\);\hfill\textup{[linearity]}
\item\label{axiom:prev:bounds} \(\linprev(\measurement{A})\geq\min\spectrum{A}\).\hfill\textup{[accepting sure gains]}
\end{enumerate}
We can associate with every coherent lower prevision~\(\lowprev\) the following closed\footnote{\dots\ in the weak\({}^\star\) topology (of point-wise convergence) \cite{Lesot_coherent_2020}.} convex set of dominating coherent previsions
\(\credalset_{\lowprev}\coloneqq\set{\linprev\given\group{\forall\measurement{A}\in\measurements}\linprev(\measurement{A})\geq\lowprev(\measurement{A})}\), called the associated \emph{credal set}.
A straightforward application of the Hahn\textendash Banach Theorem shows that a real bounded functional~\(\lowprev\) is a coherent lower prevision if and only if it's the lower envelope of the credal set~\(\credalset_{\lowprev}\), or in other words, if \(\lowprev(\measurement{A})=\min\set{\linprev(\measurement{A})\given\linprev\in\credalset_{\lowprev}}\) for all~\(\measurement{A}\in\measurements\); see \cite[Props.~2.3 and 2.4]{ITIP} and \cite{Lesot_coherent_2020}.
Coherent lower previsions are equivalent to coherent sets of desirable measurements up to border behaviour, and credal sets are equivalent to coherent lower previsions, so all three types of models can be used to describe Your beliefs.

\subsection{Density operators}
In the standard, probabilistic, framework for dealing with uncertainty in quantum mechanics, the (epistemic) uncertainty about a system's state~\(\uket\) is usually modelled by a (positive) probability mass function~\(p_1,\dots,p_r\) over possible states~\(\gket[1],\dots,\gket[r]\).
Such an `uncertain state' is called a \emph{mixed state}, and corresponds to a \emph{density operator}~\(\density\coloneqq\sum_{k=1}^rp_k\gket[k]\gbra[k]\).
The set of all such density operators is denoted by~\(\densities\).
The following is then a basic result; see \cite[Thm.~2.5]{Nielsen}.

\begin{proposition}\label{prop:dens}
A linear operator \(\density\) on~\(\hilbertspace\) is a density operator if and only if it's a Hermitian operator such that \(\trace{\density}=1\) and \(\density\geq\zero\).\footnote{The trace~\(\trace{\measurement{A}}\) of the Hermitian operator~\(\measurement{A}\) is the sum of its eigenvalues. Given an orthonormal basis~\(\set{\gket[1],\dots,\gket[n]}\) of~\(\hilbertspace\), the trace can also be written as \(\trace{\measurement{A}}=\sum_{k=1}^n\gbra[k]\measurement{A}\gket[k]\).}
\end{proposition}
\noindent According to \emph{Born's rule}, the standard probabilistic postulate in quantum mechanics, the expected outcome of a measurement~\(\measurement{A}\) is then \(\expectation[\density](\measurement{A})=\trace{\density\measurement{A}}\).
While the sets of desirable measurements approach doesn't start from the assumption that there are probabilities in quantum mechanics, nor relies on anything remotely related to Born's rule, it does allow us to recover density operators and the trace formula as a special case, as formalised in the following result; see \cite[p.~19]{Benavoli_2016} and \cite{devos2025:quantum:decision}.

\begin{theorem}\label{thm:linear:prevision:representation}
A functional~\(\linprev\) on~\(\measurements\) is a coherent prevision if and only if there's a (then \emph{unique}) density operator \(\density[\linprev]\) in \(\densities\) such that \(\linprev(\measurement{A})=\trace{\density[\linprev]\measurement{A}}\) for all~\(\measurement{A}\in\measurements\).
\end{theorem}
\noindent Therefore, as Your beliefs about \(\uket\) can be described by a credal set~\(\credalset_{\lowprev}\), we can equivalently describe them using a convex closed\footnote{\dots\ in the topology that's isomorphic to the weak\({}^\star\) topology on the space of linear functionals.} set of density operators:
\begin{align}
\densities_{\lowprev}
\coloneqq&\set{\density[\linprev]\given\linprev\in\credalset_{\lowprev}}\notag\\
=&\set{\density\in\densities\given\group{\forall\measurement{A}\in\measurements}\trace{\density\measurement{A}}\geq\lowprev(\measurement{A})},
\label{eq:from:lowprev:to:densities}
\end{align}
and \(\lowprev(\measurement{A})=\min\set{\trace{\density\measurement{A}}\given\density\in\densities_{\lowprev}}\) for all~\(\measurement{A}\in\measurements\).

\section{Conditioning on a subspace}\label{sec:conditioning}
In accordance with the postulates of quantum mechanics, You can interact with a quantum system through measurements.
The existing (and quite fruitful) ideas and techniques for \emph{conditioning} of sets of desirable gambles and dealing with symmetry \cite{walley2000,ITIP,miranda2010,Gertenjasper,Gert,lower_previsions} can also be applied in the present specific context of desirable measurements.
We intend to show here that they lead naturally to an extension of \emph{Lüders conditioning} \cite{luders1950}.
For reasons of conciseness, we'll translate the existing argumentation for desirable gambles directly to the desirable measurements context, without giving an explicit account of the former; we'll leave it to the reader to connect the dots explicitly, should they wish to do so.

Assume that You know, in some way, that the uncertain state \(\uket\) belongs to some subspace \(\subspace\subseteq\hilbertspace\).
This could \emph{for instance} be an eigenspace \(\eigspace_\eigval\) of some measurement \(\measurement{A}\), because You consider the system immediately after having performed that measurement and having observed the outcome \(\eigval\in\spectrum{A}\).
We denote the Hermitian (orthogonal) projection operator onto \(\subspace\) by \(\projection[\subspace]\).

The first question we are going to look at is how we can make sure that Your set of desirable measurements \(\desirset\) reflects this knowledge, besides any other beliefs You might have?
Observe that \(\uket\in\subspace\) is equivalent to \(\projection[\subspace]\uket=\uket\), and therefore also to
\[
\utility{A}\group{\uket}
=\utility{A}\group[\big]{\projection[\subspace]\uket}
=\utilitypure{\projection[\subspace]\measurement{A}\projection[\subspace]}(\uket)
\text{ for all \(\measurement{A}\in\measurements\)}.
\]
Therefore, You will judge any measurement \(\measurement{A}\) for which \(\projection[\subspace]\measurement{A}\projection[\subspace]=\zero\) to be \emph{equivalent to the null measurement \(\zero\)}, simply because the information that \(\uket\in\subspace\) makes You sure that the reward resulting from the measurement \(\measurement{A}\) will be zero.
We call any such measurement \emph{indifferent} to You.
Consider the linear projector \(\proj[\subspace]\colon\measurements\to\measurements\) defined by \(\proj[\subspace](\measurement{A})\coloneqq\projection[\subspace]\measurement{A}\projection[\subspace]\), then its kernel is the linear subspace of indifferent measurements \(\indifset[\subspace]\coloneqq\set{\measurement{A}\in\measurements\given\proj[\subspace](\measurement{A})=\zero}\).
It's important to observe at this point that
\begin{equation}\label{eq:ordering:and:projection}
\measurement{A}\weakgeq\zero\then\projof[\subspace]{A}\weakgeq\zero,
\text{ for all \(\measurement{A}\in\measurements\).}
\end{equation}
Interestingly, the elements of \(\indifset[\subspace]\) are the \emph{only} measurements that You definitely should be indifferent to, based on the knowledge that \(\uket\in\subspace\).
To see why, consider any \(\measurement{A}\in\measurements\) such that \(\proj[\subspace](\measurement{A})\neq\zero\), then there's some \(\gket\in\statespace\) for which \(0\neq\utilitypure{\projection[\subspace]\measurement{A}\projection[\subspace]}(\gket)=\utility{A}(\projection\gket)\), or in other words, there's some \(\fket\coloneqq\projection\gket\in\subspace\) for which \(\utility{A}(\fket)\neq0\), making it difficult to argue that You should be indifferent to \(\measurement{A}\).

The converse also holds: if You're indifferent to all measurements in \(\indifset[\subspace]\), then You must believe that, necessarily, \(\uket\in\subspace\).
Indeed, assume towards contradiction that You believe it possible that \(\uket\) equals some \(\fket\notin\subspace\), and consider the measurement \(\measurement{A}\coloneq-\projection[\subspace^\orth]\fket\fbra\projection[\subspace^\orth]\)\footnote{We denote the orthogonal complement of a subspace \(\subspace\) by \(\subspace^\orth\), and then we have that \(\projection[\subspace^\orth]=\identity-\projection[\subspace]\).}, which gives You a negative reward if \(\fket\) obtains and a non-positive reward otherwise.
Since, clearly, \(\measurement{A}\in\indifset[\subspace]\), You're then indifferent to, and therefore have fair price zero for,\footnote{It's an easy exercise to show that if \(\desirset+\indifset\subseteq\desirset\) for some linear space \(\indifset\), then \(\lprice[\desirset]\group{\measurement{A}}=\uprice[\desirset]\group{\measurement{A}}=0\) for all~\(\measurement{A}\in\indifset\) so Your fair price for any indifferent measurement is zero; see also \cref{eq:desirability:indifference}.} a loss You deem possible, without any possibility of gain; this is unreasonable.

Your existing beliefs, as captured in a coherent set of desirable measurements \(\desirset\) in \(\measurements\) are in accordance with Your knowledge that \(\uket\in\subspace\) if they satisfy
\begin{equation}\label{eq:desirability:indifference}
\desirset+\indifset[\subspace]\subseteq\desirset,
\end{equation}
which expresses that the desirability of any measurement remains unaffected when adding indifferent measurements to it.\footnote{For a nice justification of this condition, see \cite{quaeghebeur2015:statement}.}
We call any set of desirable measurements that satisfies this condition \emph{\(\indifset[\subspace]\)-compatible}.

Instead of working with indifferent measurements, we can also translate the knowledge that \(\uket\in\subspace\) into the framework of desirability by looking at the measurements~\(\measurement{A}\) that are \emph{invariant} under \(\proj[\subspace]\), in the sense that \(\measurement{A}=\proj[\subspace]\group{\measurement{A}}=\projection[\subspace]\measurement{A}\projection[\subspace]\).
Under this invariance condition, \(\measurement{A}\in\desirset\ifandonlyif\projection[\subspace]\measurement{A}\projection[\subspace]\in\desirset\).
We'll call any set of desirable measurements in \(\measurements\) \emph{\(\subspace\)-focused} if
\begin{equation}\label{eq:desirability:focused}
\group{\forall\measurement{A}\in\measurements}
\group{\measurement{A}\in\desirset\ifandonlyif\proj[\subspace]\group{\measurement{A}}\in\desirset}.
\end{equation}
We now show that these two conditions~\labelcref{eq:desirability:indifference,eq:desirability:focused} are equivalent under coherence, thus answering the question of how to express that Your set of desirable measurements \(\desirset\) incorporates Your knowledge that \(\uket\in\subspace\).

\begin{proposition}\label{prop:indifference:subspace}
Any coherent set of desirable measurements in \(\measurements\) is \(\indifset[\subspace]\)-compatible if and only if it's \(\subspace\)-focused.
\end{proposition}

\begin{proof}
For necessity, assume that \(\desirset\) is \(\indifset\)-compatible.
For any~\(\measurement{A}\in\measurements\), \(\proj[\subspace](\measurement{A}-\projof[\subspace]{A})=\projof[\subspace]{A}-\projof[\subspace]{A}=\zero\), whence \(\measurement{A}-\projof[\subspace]{A}\in\indifset\) and \(\proj[\subspace](\measurement{A})-\measurement{A}\in\indifset\).
But now \(\measurement{A}\in\desirset\) implies that \(\projof[\subspace]{A}=\measurement{A}+\group{\projof[\subspace]{A}-\measurement{A}}\in\desirset+\indifset\subseteq\desirset\); similarly, \(\projof[\subspace]{A}\in\desirset\) implies that \(\measurement{A}=\projof{A}+\group{\measurement{A}-\projof[\subspace]{A}}\in\desirset+\indifset\subseteq\desirset\).

For sufficiency, assume that \(\desirset\) is \(\projection[\subspace]\)-focused.
For any \(\measurement{A}\in\desirset\) and \(\measurement{B}\in\indifset\), \(\proj(\measurement{A}+\measurement{B})=\proj[\subspace](\measurement{A})+\measurement{0}=\projof[\subspace]{A}\in\desirset\), so \(\measurement{A}+\measurement{B}\in\desirset\).
\end{proof}

\section{Abstract updating results}\label{sec:abstract}
Before continuing with the problem of conditioning on a subspace, we're going to lift its formulation to a more abstract level, solve it there, and then in \Cref{sec:back:to:conditioning} translate the solution back to the original context.
This will then also allow us in \Cref{sec:general:conditioning} to deal effortlessly with more general types of conditioning.

Consider any linear operator \(\proj\colon\measurements\to\measurements\) that's a \emph{projection}, in the sense that \(\proj\circ\proj=\proj\).
Inspired by the condition~\eqref{eq:ordering:and:projection}, we'll require in addition that
\begin{equation}\label{eq:ordering:and:projection:general}
\measurement{A}\weakgeq\zero\then\projof{A}\weakgeq\zero,
\text{ for all \(\measurement{A}\in\measurements\).}
\end{equation}
Crucially for what follows, we'll assume that its kernel \(\projindifset\coloneqq\set{\measurement{A}\in\measurements\given\projof{A}=\zero}\) is Your linear space of indifferent measurements.

\subsection{Representing the indifference}
The first question we're going to address, is how to express that a given coherent set of desirable measurements \(\desirset\) \emph{incorporates}, possibly along with other beliefs, Your indifference to the measurements in \(\projindifset\).

Inspired by the discussion in the previous section, we call any coherent set of desirable measurements \(\desirset\) \emph{\(\projindifset\)-compatible} if
\begin{equation}\label{eq:desirability:indifference:general}
\desirset+\projindifset\subseteq\desirset,
\end{equation}
and \emph{\(\rng(\proj)\)-focused} if
\begin{equation}\label{eq:desirability:focused:general}
\group{\forall\measurement{A}\in\measurements}
\group{\measurement{A}\in\desirset\ifandonlyif\projof{A}\in\desirset},
\end{equation}
where \(\rng(\proj)\coloneqq\proj\group{\measurements}=\set{\measurement{A}\in\measurements\given\projof{A}=\measurement{A}}\) is the \emph{range} of the projection~\(\proj\).
\Cref{prop:indifference:subspace} generalises at once to the following result.

\begin{proposition}\label{prop:indifference:subspace:general}
Any coherent set of desirable measurements in \(\measurements\) is \(\projindifset\)-compatible if and only if it's \(\rng(\proj)\)-focused.
\end{proposition}

\noindent The \(\rng(\proj)\)-focused condition hints at the idea that, as far as \(\projindifset\)-indifference is concerned, all the action actually takes place in the linear subspace \(\rng(\proj)\) of \(\measurements\), and that we can represent this action using coherent sets of desirable measurements in that lower-dimensional subspace.
To formalise this idea, we need a few extra notations and definitions.

We'll need maps that allow us to turn measurements in \(\measurements\) into measurements in \(\rng(\proj)\) and vice versa.
We let \(\reduce\colon\measurements\to\rng(\proj)\colon\measurement{A}\mapsto\projof{A}\), so \(\reduceof{A}\) is the restriction of \(\projof{A}\) to the linear subspace \(\rng(\proj)\).
On the other hand, we consider the canonical embedding \(\extend\colon\rng(\proj)\to\measurements\colon\measurement{C}\mapsto\measurement{C}\).
Keep in mind that \(\extend(\reduceof{A})=\projof{A}\) and \(\reduce(\extendof{C})=\measurement{C}\) for all \(\measurement{A}\in\measurements\) and \(\measurement{C}\in\rng(\proj)\),

If we want to define coherence for sets of desirable measurements in \(\rng(\proj)\), we also need to introduce respective counterparts \(\rngstronggt\) and \(\rngweakgeq\) on \(\rng(\proj)\) of the vector orderings \(\stronggt\) and \(\weakgeq\) on \(\measurements\):\footnote{In these expressions, we've identified \(\extendof{C}\) and \(\measurement{C}\). These orderings are isomorphic copies of the so-called \emph{quotient orderings} of \(\stronggt\) and \(\weakgeq\) on the quotient space \(\measurements/\projindifset\) of \(\measurements\) by the kernel~\(\projindifset\) of \(\proj\), the latter being linearly isomorphic to \(\rng(\proj)\). Moreover, \(\rngposdefmeasurements=\proj\group{\posdefmeasurements}\) and \(\rngpossemidefmeasurements=\proj\group{\possemidefmeasurements}\). Further motivation for these definitions will be given in \Cref{sec:back:to:conditioning}.}
\begin{equation*}
\left.
\begin{aligned}
\measurement{C}\rngstronggt\zero&\ifandonlyif\group{\exists\measurement{D}\in\projindifset}\,\measurement{C}\stronggt\measurement{D}\\
\measurement{C}\rngweakgeq\zero&\ifandonlyif\group{\exists\measurement{D}\in\projindifset}\,\measurement{C}\weakgeq\measurement{D}
\end{aligned}
\right\}
\text{ for any \(\measurement{C}\in\rng(\proj)\)},
\end{equation*}
leading to the convex cones \(\rngposdefmeasurements\) and \(\rngpossemidefmeasurements\) in \(\rng(\proj)\) in the usual manner.
The requirement~\eqref{eq:ordering:and:projection:general} has interesting consequences for \(\rngweakgeq\).

\begin{proposition}\label{prop:ordering:and:projection:general}
\(\rngpossemidefmeasurements=\possemidefmeasurements\cap\rng(\proj)\).
\end{proposition}

\begin{proof}
Consider any \(\measurement{A}\in\rng(\proj)\), so \(\measurement{A}=\projof{A}\).
It's clearly enough to show that \(\measurement{A}\rngweakgeq\zero\then\measurement{A}\weakgeq\zero\).
So assume that there's some \(\measurement{D}\in\projindifset\) such that \(\measurement{A}+\measurement{D}\weakgeq\zero\), then the requirement~\eqref{eq:ordering:and:projection:general} leads us to conclude that \(\measurement{A}=\projof{A}=\proj\group{\measurement{A}+\measurement{D}}\weakgeq\zero\), so we're done.
\end{proof}

We can now use the map \(\extend\) to turn \emph{any} set of desirable measurements \(\desirset\) in \(\measurements\) into an assessment of desirable measurements~\(\assessment_{\desirset}\) \emph{in the subspace \(\rng(\proj)\)} where all the relevant action is, as follows:
\begin{align}
\assessment_{\desirset}
\coloneqq&\invextend\group{\desirset}
=\set{\measurement{C}\in\rng(\proj)\given\extend(\measurement{C})\in\desirset}\notag\\
=&\desirset\cap\rng(\proj).
\label{eq:indifference:assessment:reduce}
\end{align}
If this assessment turns out to be consistent, it can be extended to a smallest --- most conservative --- coherent set of desirable measurements in \(\rng(\proj)\) using natural extension in that linear subspace [see \cref{eq:natex}]:
\begin{multline}
\projlocalcond{\desirset}
\coloneqq\natext(\assessment_{\desirset})\\
=\rngposdefmeasurements\cup\group[\big]{\posi(\assessment_{\desirset})+\rngpossemidefmeasurements},\label{eq:indifference:representation:reduce}
\end{multline}
Interestingly, this operation has an inverse when restricted to the set of \(\projindifset\)-compatible sets of desirable measurements.
To see this, we introduce a new set of desirable options \emph{in \(\measurements\)}, determined by the set of desirable options \(\projlocalcond{\desirset}\) in \(\rng(\proj)\) as follows:
\begin{multline}\label{eq:indifference:representation:extend}
\projcond{\desirset}
\coloneqq\invreduce(\projlocalcond{\desirset})\\
=\set{\measurement{A}\in\measurements\given\reduce(\measurement{A})\in\projlocalcond{\desirset}}.
\end{multline}

\begin{proposition}[Representation]\label{prop:local:representation:general}
Let \(\desirset\) be any coherent set of desirable measurements on \(\hilbertspace\), then the set of desirable measurements \(\projlocalcond{\desirset}\) in \(\rng(\proj)\) is coherent as well, and  \(\desirset\) is \(\projindifset\)-compatible if and only if \(\projcond{\desirset}=\desirset\).
\end{proposition}

\begin{proof}
As \(\projlocalcond{\desirset}\) is the natural extension of \(\assessment_{\desirset}\), we only need to prove \labelcref{axiom:desirs:consistency}, or in this case \(\assessment_{\desirset}\cap\rng(\proj)_{\leq0}=\emptyset\), to establish the coherence of \(\projlocalcond{\desirset}\).
By combining \cref{eq:indifference:assessment:reduce} with \Cref{prop:ordering:and:projection:general}, we find that
\begin{align*}
\assessment_{\desirset}\cap\rng(\proj)_{\leq0}
&=\desirset\cap\rng(\proj)\cap\rng(\proj)_{\leq0}\\
&=\desirset\cap\rng(\proj)\cap\group[\big]{\negsemidefmeasurements\cap\rng(\proj)}\\
&=\underset{=\emptyset\textrm{ by \labelcref{axiom:desirs:consistency}}}{\underbrace{\group{\desirset\cap\negsemidefmeasurements}}}\cap\rng(\proj)=\emptyset.
\end{align*}
For the proof of the second statement, we begin with necessity.
Assume that \(\desirset\) is \(\projindifset\)-compatible.
By \Cref{prop:indifference:subspace:general}, \(\desirset\) is \(\rng(\proj)\)-focused, meaning that \(\measurement{A}\in\desirset\) if and only if \(\projof{A}=\extend(\reduce(\measurement{A}))\in\desirset\), for all \(\measurement{A}\in\measurements\).
So for any \(\measurement{A}\in\desirset\), it follows that \(\reduce(\measurement{A})\in\projlocalcond{\desirset}\), so \(\measurement{A}\in\projcond{\desirset}\).
Hence, \(\desirset\subseteq\projcond{\desirset}\).
Conversely, consider any \(\measurement{A}\in\projcond{\desirset}\), then \(\reduce(\measurement{A})\in\projlocalcond{\desirset}\), so it follows from \cref{eq:indifference:representation:reduce} that there are two possibilities.
The first is that \(\reduce(\measurement{A})=\measurement{B}+\measurement{C}\), with \(\measurement{B},\measurement{C}\in\rng(\proj)\), \(\measurement{B}=\extendof{B}\in\desirset\) and \(\measurement{C}\in\rngpossemidefmeasurements\), and therefore, by \Cref{prop:ordering:and:projection:general}, \(\extendof{C}=\measurement{C}\geq0\).
But then \(\projof{A}=\extend(\reduce(\measurement{A}))=\extend(\measurement{B}+\measurement{C})=\measurement{B}+\measurement{C}\geq\measurement{B}\in\desirset\), whence \(\projof{A}\in\desirset\) [use \labelcref{axiom:desirs:monotonicity}].
This implies that \(\measurement{A}\in\desirset\).
The second possibility is that \(\reduce(\measurement{A})\in\rngposdefmeasurements\), so there's some \(\measurement{D}\in\projindifset\) such that \(\projof{A}+\measurement{D}>0\), and therefore, by \labelcref{axiom:desirs:background}, \(\projof{A}+\measurement{D}\in\desirset\).
Since also \(\measurement{A}-\projof{A}\in\projindifset\), and therefore \(\measurement{A}-\projof{A}-\measurement{D}\in\projindifset\), we find that \(\measurement{A}=\group{\projof{A}+\measurement{D}}+\group{\measurement{A}-\projof{A}-\measurement{D}}\in\desirset+\projindifset\subseteq\desirset\).
We see that, indeed, \(\projcond{\desirset}\subseteq\desirset\).

For sufficiency, observe that always \(\projcond{\desirset}+\projindifset\subseteq\projcond{\desirset}\), because \(\reduce(\measurement{A}+\measurement{B})=\reduce(\measurement{A})\) for all \(\measurement{A}\in\measurements\) and \(\measurement{B}\in\projindifset\).
\end{proof}
\noindent Any coherent and \(\projindifset\)-compatible set of desirable measurements on \(\hilbertspace\) can therefore be represented by a coherent set of desirable measurements on the typically lower-dimensional subspace \(\subspace\).
This result has a converse, which will be helpful for the argumentation further on.
\begin{proposition}\label{prop:original:representation}
Consider any coherent set of desirable measurements \(\desirset[o]\) in \(\rng(\proj)\), then
\begin{equation*}
\desirset
\coloneqq\invreduce(\desirset[o])
=\set{\measurement{A}\in\measurements\given\reduce(\measurement{A})\in\desirset[o]}
\end{equation*}
is \(\rng(\proj)\)-focused and coherent, and \(\desirset[o]=\projlocalcond{\desirset}\).
\end{proposition}

\begin{proof}
The set \(\desirset\) is clearly \(\projindifset\)-compatible, since for any \(\measurement{A}\in\desirset\) and \(\measurement{B}\in\projindifset\),
\begin{align*}
\reduce(\measurement{A}+\measurement{B})
&=\reduce(\measurement{A})+\reduce(\extend(\reduce(\measurement{B})))\\
&=\reduce(\measurement{A}+\proj(\measurement{B}))
=\reduce(\measurement{A})\in\desirset[o],
\end{align*}
and therefore \(\measurement{A}+\measurement{B}\in\desirset\).
For coherence, \labelcref{axiom:desirs:additivity,axiom:desirs:scaling} follow from the linearity of \(\reduce\).
For \labelcref{axiom:desirs:background}, consider any \(\measurement{A}\stronggt\zero\).
Since \(\measurement{A}-\projof{A}\in\projindifset\) and \(\measurement{A}=\projof{A}+\group{\measurement{A}-\projof{A}}\), it follows that \(\reduceof{A}\rngstronggt\zero\), and therefore also \(\reduceof{A}\in\desirset[o]\), by \labelcref{axiom:desirs:background}.
Hence, \(\measurement{A}\in\desirset\).
For \labelcref{axiom:desirs:monotonicity}, consider any \(\measurement{A}\weakgeq\zero\), then also \(\projof{A}\weakgeq\zero\) by the requirement~\eqref{eq:ordering:and:projection:general}.
Hence, \(\reduceof{A}\rngweakgeq\zero\), which combined with the linearity of \(\reduce\) ensures monotonicity.
Finally, for \labelcref{axiom:desirs:consistency}, assume that \(\measurement{A}\weakleq\zero\), then also \(\projof{A}\weakleq\zero\) by the requirement~\eqref{eq:ordering:and:projection:general}.
It follows that \(\reduceof{A}\rngweakleq\zero\), and therefore \(\reduceof{A}\notin\desirset[o]\), by \labelcref{axiom:desirs:consistency}.
Hence, indeed, \(\measurement{A}\notin\desirset\).

For the second statement, consider any \(\measurement{C}\in\rng(\proj)\), then \(\measurement{C}\in\assessment_{\desirset}\) if and only if \(\extend(\measurement{C})\in\desirset\), which is in turn equivalent to \(\measurement{C}=\reduce(\extend(\measurement{C}))\in\desirset[o]\).
Therefore, \(\desirset[o]=\assessment_{\desirset}\), which implies that, since \(\desirset[o]\) is coherent, \(\desirset[o]=\natext(\desirset[o])=\natext(\assessment_{\desirset})=\projlocalcond{\desirset}\).
\end{proof}

\noindent We're of course mainly interested in what these expressions become when we start out with a set of desirable measurements \(\desirset\) that's coherent.

\begin{proposition}\label{prop:conditioning:general}
Let \(\desirset\) be any coherent set of desirable measurements on \(\hilbertspace\), then
\begin{align}
\projlocalcond{\desirset}
&=\rngposdefmeasurements\cup\group[\big]{\desirset\cap\rng(\proj)}
\label{eq:conditioning:coherent:general:local}\\
\projcond{\desirset}
&=\projnewbackground\cup\invproj\group{\desirset}
\label{eq:conditioning:coherent:general:first}\\
&=\projnewbackground\cup\set{\measurement{A}\in\measurements\given\projof{A}\in\desirset}
\label{eq:conditioning:coherent:general:second}
\end{align}
with \(\projnewbackground\coloneqq\set{\measurement{A}\in\measurements\given\group{\exists\measurement{D}\in\projindifset}\,\projof{A}\stronggt\measurement{D}}\).
\end{proposition}

\begin{proof}
It's clear from \cref{eq:indifference:representation:extend} that \cref{eq:conditioning:coherent:general:first,eq:conditioning:coherent:general:second} follow directly from \cref{eq:conditioning:coherent:general:local}, so we concentrate on the latter.
First, combine the result in \cref{eq:indifference:assessment:reduce} with \cref{prop:ordering:and:projection:general} to find that
\begin{multline*}
\group[\big]{\desirset\cap\rng(\proj)}+\rngpossemidefmeasurements\\
\begin{aligned}
&=\group[\big]{\desirset\cap\rng(\proj)}+\group[\big]{\possemidefmeasurements\cap\rng(\proj)}\\
&\subseteq\group{\desirset+\possemidefmeasurements}\cap\rng(\proj)\\
&\subseteq\desirset\cap\rng(\proj),
\end{aligned}
\end{multline*}
where the first inclusion follows because \(\rng(\proj)\) is a linear space, and the second inclusion follows from \labelcref{axiom:desirs:monotonicity}.
The converse inequality holds because \(\zero\rngweakgeq\zero\), so we find that \(\group{\desirset\cap\rng(\proj)}+\rngpossemidefmeasurements=\desirset\cap\rng(\proj)\), and therefore also that \(\posi\group{\group{\desirset\cap\rng(\proj)}+\rngpossemidefmeasurements}=\desirset\cap\rng(\proj)\).
Now take this equality back to \cref{eq:indifference:representation:reduce}, with \(\assessment_{\desirset}=\desirset\cap\rng(\proj)\).
\end{proof}

\subsection{Updating Your beliefs}
Above, we took some pains to investigate what a coherent set of desirable measurements \(\desirset\) must look like to reflect, besides perhaps other beliefs, an indifference assessment \(\projindifset\) associated with a projection operator \(\proj\) satisfying requirement~\eqref{eq:ordering:and:projection:general}: it must be \(\projindifset\)-compatible, or equivalently, \(\rng(\proj)\)-focused.
We then also investigated how it can be represented in the lower-dimensional subspace \(\rng(\proj)\), where all the relevant action then is.

We now turn to the more involved question of how to combine Your initial beliefs about the system's state \(\uket\), as captured in a coherent set of desirable measurements \(\desirset\), with a new indifference assessment of this type.

Because Your updated beliefs must reflect the indifference assessments, they must be \(\projindifset\)-compatible, which effectively makes sure that they can be represented by some coherent set of desirable measurements in the linear subspace \(\rng(\proj)\) of \(\measurements\).
Its compatibility with the initial beliefs ---as captured in the set \(\desirset\) of desirable measurements on~\(\hilbertspace\) --- is expressed by the fact that \emph{we only take those measurements in \(\rng(\measurements)\) that were initially desirable}, so Your \emph{updated set of desirable measurements} in the reduced space \(\rng(\proj)\) is given by
\begin{align}
\projlocalcond{\desirset}
&=\natext(\desirset\cap\rng(\proj))\notag\\
&=\rngposdefmeasurements\cup\group[\big]{\desirset\cap\rng(\proj)};
\label{eq:conditioning:coherent:general:local:simpler}
\end{align}
see \Cref{prop:conditioning:general} for the second equality.
\(\projlocalcond{\desirset}\) is coherent in \(\rng(\proj)\) by \Cref{prop:local:representation:general}.

However, as we need the entire state space, and therefore need to work with all measurements, to deal with any possible further evolution of the quantum system, we can use the \(\rng(\proj)\)-focused set of desirable measurements \(\projcond{\desirset}\) \emph{in \(\measurements\)} that corresponds to the coherent set \(\projlocalcond{\desirset}\) of desirable measurements in \(\rng(\proj)\):
\begin{align}
\projcond{\desirset}
&=\invreduce(\projlocalcond{\desirset})\notag\\
&=\projnewbackground\cup\set{\measurement{A}\in\measurements\given\projof{A}\in\desirset};
\label{eq:conditioning:coherent:general:simpler}
\end{align}
see \Cref{prop:conditioning:general} for the second equality.
By \Cref{prop:original:representation}, this set is coherent in \(\measurements\), so we've achieved our goal: we've turned Your initial coherent set of desirable measurements \(\desirset\) into a new coherent and \(\projindifset\)-compatible set of desirable measurements \(\projcond{\desirset}\).

\subsection{Duality}\label{sec:duality}
When Your beliefs are described by a coherent prevision \(\linprev\) on \(\measurements\),
\begin{equation*}
\desirset[\linprev]
\coloneqq\set{\measurement{A}\in\measurements\given\linprev(\measurement{A})>0}
\end{equation*}
is a corresponding coherent set of desirable measurements on \(\hilbertspace\),
because \(\linprev=\lprice[{\desirset[\linprev]}]=\uprice[{\desirset[\linprev]}]\).
Updating \(\desirset[\linprev]\) with the indifference assessment~\(\projindifset\) leads to
\begin{equation*}
\projcond{\desirset[\linprev]}
=\projnewbackground\cup\set{\measurement{A}\in\measurements\given\linprev(\projof{A})>0}.
\end{equation*}
The linearity of the projection \(\proj\) guarantees that this updated coherent set of desirable measurements corresponds to the \emph{updated coherent prevision} \(\linprev(\bolleke\vert\rng(\proj))\coloneqq\lprice[{\projcond{\desirset[\linprev]}}]=\uprice[{\projcond{\desirset[\linprev]}}]\) on \(\measurements\), with
\begin{align}
\linprev\group{\measurement{A}\vert\rng(\proj)}
&=\inf\set{\con\in\reals\given\con\identity-\measurement{A}\in\projcond{\desirset}}
\notag\\
&=\frac{\linprev\group{\projof{A}}}{\linprev\group{\projof{I}}}
\text{ for all~\(\measurement{A}\in\measurements\)}.
\label{eq:prevision:update}
\end{align}

\section{Back to conditioning on a subspace}\label{sec:back:to:conditioning}
We're now ready to apply these general results to the specific case we'd begun studying in \Cref{sec:conditioning}, where the new information that the system's state \(\uket\) belongs to some subspace \(\subspace\) is represented by the linear space \(\indifset[\subspace]\) of indifferent measurements that's the kernel of the projection operator \(\proj[\subspace]=\projection[\subspace]\bolleke\projection[\subspace]\) on \(\hilbertspace\), so \(\indifset[\subspace]=\projindifset[\subspace]\).

Before going through the details, it's useful and interesting to point out that the range \(\rng(\proj[\subspace])\) of the projection \(\proj[\subspace]\) is linearly isomorphic to the linear space \(\measurements(\subspace)\) of all the measurements on the subspace \(\subspace\) itself, through the linear isomorphism \(\rng(\proj[\subspace])\to\measurements(\subspace)\) mapping any \(\measurement{A}\) in \(\rng(\proj[\subspace])\) to the measurement \(\measurement{C}\) in \(\measurements(\subspace)\) uniquely determined by \(\fbra\measurement{C}\fket=\fbra\measurement{A}\fket\) for all \(\fket\in\subspace\).
The inverse linear isomorphism \(\measurements(\subspace)\to\rng(\proj[\subspace])\) is given by \(\measurement{C}\mapsto\projection[\subspace]\measurement{C}\projection[\subspace]\), with some abuse of notation.

Also, if \(\measurement{A}=\projof{A}\rngweakgeq[\subspace]\zero\), then there's some \(\measurement{D}\in\indifset[\subspace]\) such that \(\measurement{A}+\measurement{D}\weakgeq\zero\), and this leads to the following chain of implications
\begin{align*}
\measurement{A}+\measurement{D}\weakgeq\zero
&\then\group{\forall\fket\in\subspace}\,\fbra\measurement{A}+\measurement{D}\fket\geq0\\
&\then\group{\forall\fket\in\subspace}\fbra\projection[\subspace]\group{\measurement{A}+\measurement{D}}\projection[\subspace]\fket\geq0\\
&\then\group{\forall\fket\in\subspace}\fbra\projection[\subspace]\measurement{A}\projection[\subspace]\fket\geq0\\
&\then\group{\forall\fket\in\subspace}\fbra\measurement{C}\fket\geq0,
\end{align*}
implying that the isomorphic copy \(\measurement{C}\) of \(\measurement{A}\) is positive semidefinite, so \(\measurement{C}\geq0\) in \(\measurements(\subspace)\).
Since the converse implication holds trivially, we see that the ordering \(\rngweakgeq[\subspace]\) on \(\rng(\proj[\subspace])\) is an isomorphic copy of the ordering \(\weakgeq\) associated with positive semidefiniteness on \(\measurements(\subspace)\); and similarly, the ordering \(\rngstronggt[\subspace]\) on \(\rng(\proj[\subspace])\) is an isomorphic copy of the positive definiteness ordering \(\stronggt\) on \(\measurements(\subspace)\).

We can now (re)turn to the question of how to combine Your initial beliefs about the system's state \(\uket\), as captured in a coherent set of desirable measurements \(\desirset\) in \(\measurements\), with the new knowledge that the system's state \(\uket\) belongs to some subspace \(\subspace\).

Applying the results of the previous section to this particular case, taking into account the above-mentioned isomorphisms and the correspondences between orderings, we find that since Your updated beliefs must incorporate the indifference assessments, they must be \(\indifset[\subspace]\)-compatible, which effectively makes sure that they can be represented by some coherent set of desirable measurements in the linear space \(\measurements(\subspace)\) of all measurements on the subspace \(\subspace\).
Their compatibility with Your initial beliefs, which are captured in the set \(\desirset\) of desirable measurements in \(\measurements\), is expressed by the fact that their extension \(\projection[\subspace]\measurement{C}\projection[\subspace]\) to the original space \(\hilbertspace\) was initially desirable.
In other words, \cref{eq:conditioning:coherent:general:local:simpler} now turns into, with obvious notations,
\begin{align}
\localcond{\desirset}
&=\natext\group{\set{\measurement{C}\in\measurements(\subspace)\given\projection[\subspace]\measurement{C}\projection[\subspace]\in\desirset}}\notag\\
&=\posdefmeasurements(\subspace)\cup\set{\measurement{C}\in\measurements(\subspace)\given\projection[\subspace]\measurement{C}\projection[\subspace]\in\desirset}
\label{eq:conditioning:coherent:local:subspace}
\end{align}
as Your updated set of desirable measurements in \(\measurements(\subspace)\); it's coherent by an application of \Cref{prop:local:representation:general}.

However, any subsequent evolution of the quantum system will typically have to be described in the whole state space \(\statespace\), as its dynamics may perfectly well make its state move outside the subspace \(\subspace\) again.
\Cref{prop:original:representation} now tells us, {\itshape mutatis mutandis}, that we may find this description in \(\statespace\) by considering the \(\subspace\)-focused set of desirable measurements in \(\measurements\) given by the following appropriately transformed version of \cref{eq:conditioning:coherent:general:simpler}:
\begin{equation}\label{eq:conditioning:coherent:subspace}
\cond{\desirset}
=\newbackground\cup\set{\measurement{A}\in\measurements\given\projection[\subspace]\measurement{A}\projection[\subspace]\in\desirset},
\end{equation}
where \(\newbackground\coloneqq\set{\measurement{A}\in\measurements\colon(\exists\measurement{D}\in\indifset[\subspace])\projection[\subspace]\measurement{A}\projection[\subspace]>\measurement{D}}\), or equivalently, \(\newbackground=\proj(\posdefmeasurements)+\indifset[\subspace]\).
\Cref{prop:original:representation} also ensures that this set is coherent, so we've achieved our goal of turning the initial coherent set of desirable measurements \(\desirset\) in \(\measurements\) into a new coherent and \(\indifset[\subspace]\)-compatible set of desirable measurements \(\cond[\subspace]{\desirset}\).

This rule for conditioning closely resembles the updating rule introduced in \cite{Benavoli_2016} for conditioning, but there are important differences.
First, our interpretation of a set of desirable measurements isn't contingent on a bookmaker preparing and then measuring the system in certain way.
Secondly, we don't restrict ourselves to the case where the conditioning event is a pure state or a collection of orthogonal pure states, but we allow for conditioning on any subspace.
Finally, our rule is slightly different as we impose more general coherence axioms than the ones used in \cite{Benavoli_2016}, which leave room for imposing indifference assessments and therefore give a natural interpretation to conditioning on a subspace.

On the other hand, \cref{eq:conditioning:coherent:subspace} is also a generalisation of \emph{Lüders' conditioning rule} in quantum mechanics.
To substantiate this claim, we specialise the results of \Cref{sec:duality} to the present case where \(\proj=\proj[\subspace]\).
We look in particular at \cref{eq:prevision:update}, which transforms into \(\linprev\group{\measurement{A}\vert\subspace}=\linprev\group{\projof[\subspace]{A}}/\linprev\group{\projof[\subspace]{I}}\) for all~\(\measurement{A}\in\measurements\).
We can now concentrate on the density operators \(\density\) and \(\density[\subspace]\) that correspond to the coherent previsions \(\linprev\) and \(\linprev(\cdot\vert\subspace)\) as determined in \Cref{thm:linear:prevision:representation}.
For all \(\measurement{A}\in\measurements\),
\begin{equation*}
\trace{\density[\subspace]\measurement{A}}
=\linprev\group{\measurement{A}\vert\subspace}
=\frac{\trace{\density\projection[\subspace]\measurement{A}\projection[\subspace]}}{\trace{\density\projection[\subspace]\projection[\subspace]}}
=\frac{\trace{\projection[\subspace]\density\projection[\subspace]\measurement{A}}}{\trace{\projection[\subspace]\density\projection[\subspace]}},
\end{equation*}
leading directly to the following expressions for the updated density operator:
\begin{equation*}
\density[\subspace]
=\frac{{\projection[\subspace]\density\projection[\subspace]}}{\trace{\projection[\subspace]\group\density\projection[\subspace]}}
=\frac{\proj[\subspace]\group{\density}}{\trace{\proj[\subspace]\group{\density}}},
\end{equation*}
which correspond to Lüders' approach to conditioning a density operator \cite{luders1950}.
To summarise, if decisions about measurements are determined by a density operator \(\density\), then after learning that \(\uket\in\subspace\), the resulting decisions about measurements are now determined by the density operator \(\density[\subspace]\), which is the Lüders conditionalisation of the original density operator \(\density\) on the subspace \(\subspace\).

\begin{runexample}
Consider a quantum system with two spin \(1\) particles, where the state space describing the spin is \(\statespace\coloneqq\statespace[3]\otimes\statespace[3]\), with \(\hilbertspace[3]\coloneqq\linspanof{\set{\ket{-1},\ket{0},\ket{1}}}\) and \(\spa\) denotes the linear span.
For simplicity, we'll use the notation \(\ket{\ell,k}\coloneqq\ket{\ell}\otimes\ket{k}\).

You know that the system's total spin is zero, so that \(\uket\in\subspace\), where \(\subspace\coloneqq\linspanof{\set{\ket{-1,1},\ket{0,0},\ket{1,-1}}}\).
The linear space of indifferent measurements corresponding to this knowledge is \(\indifset[\subspace]=\{\measurement{A}\in\measurements\colon\projection[\subspace]\measurement{A}\projection[\subspace]=\zero\}\), with
\begin{equation*}
\projection[\subspace]
=\ket{-1,1}\bra{-1,1}+\ket{0,0}\bra{0,0}+\ket{1,-1}\bra{1,-1}.
\end{equation*}
The indifferent measurements are those that only differ from \(\zero\) on \(\subspace\)'s orthogonal complement \(\subspace^\orth=\linspanof{\set{\ket{-1,-1},\ket{-1,0},\ket{0,-1},\ket{0,1},\ket{1,0},\ket{1,1}}}\).
This is completely analogous to case of classical probability, where the indifferent gambles are those that only differ from \(0\) on the complement of the conditioning event.

A coherent set of desirable measurements \(\desirset\) in \(\measurements\) expresses this knowledge if it's \(\indifset[\subspace]\)-compatible, and then it can be represented by the coherent set of desirable measurements \(\localcond{\desirset}\) on \(\subspace\), which is simpler, as it only involves measurements on the \(3\)-dimensional subspace \(\subspace\).
But, since further dynamical evolution of the system will typically take it outside this subspace, we still need to consider the set of desirable measurements \(\cond{\desirset}\) on the full space \(\hilbertspace\).

Suppose now that Your initial set of desirable measurements is \(\desirset=\set{\measurement{A}\in\measurements\given\trace{\measurement{A}\density}>0}\), with
\begin{equation*}
\density
\coloneqq\frac13(\ket{-1,1}-\ket{0,0}+\ket{1,1})(\bra{-1,1}-\bra{0,0}+\bra{1,1}).
\end{equation*}
The corresponding set of density operators is then
\begin{equation*}
\densities_{\desirset}
=\set{\density\in\densities\given\group{\forall\measurement{A}\in\desirset}\trace{\density\measurement{A}}>0}
=\set{\density}.
\end{equation*}
This density operator corresponds to the system state \(\uket\) being equal to \(\nicefrac1{\sqrt3}\group{\ket{-1,1}-\ket{0,0}+\ket{1,1}}\) with probability one.
If You then learn that the system's state \(\uket\) belongs to the subspace \(\subspace=\linspanof{\ket{-1,1},\ket{0,0},\ket{1,-1}}\), then the set of desirable measurements \(\cond{\desirset}\) is given by
\begin{align*}
\cond{\desirset}
&=\set{\measurement{A}\in\measurements\given\trace{\projection[\subspace]\measurement{A}\projection[\subspace]\density}>0}\cup\newbackground\\
&=\set{\measurement{A}\in\measurements\given\trace{\measurement{A}\density[\subspace]}>0}\cup\newbackground,
\end{align*}
with \(\density[\subspace]=\projection[\subspace]\density\projection[\subspace]/\trace{\projection[\subspace]\density\projection[\subspace]}\) given by Lüders' rule.
The corresponding set of density operators is the singleton \(\densities_{\cond{\desirset}}=\set{\density[\subspace]}\), where the updated density operator
\begin{equation*}
\density[\subspace]=\frac12\group{\ket{-1,1}-\ket{0,0}}\group{\bra{-1,1}-\bra{0,0}}.
\end{equation*}
corresponds to the system state \(\uket\) being equal to \(\nicefrac1{\sqrt2}\group{\ket{-1,1}-\ket{0,0}}\) with probability one.
\end{runexample}

\section{A more general type of conditioning}\label{sec:general:conditioning}
In \Cref{sec:conditioning,sec:back:to:conditioning}, we looked at how Your beliefs about the system's state \(\uket\) can be updated when You learn that \(\uket\) belongs to some subspace \(\subspace\), for instance after doing a measurement on the system, and observing the outcome.
We may wonder, however, how to deal with the situation where You know that a measurement \(\measurement{A}\) has been performed, but You don't learn its outcome, or You learn it only partially.
You'll then be sure that the state resides in one of the eigenspaces of the measurement, but You're uncertain about which.

We're therefore now going to look at the more general problem where You come to learn that the system's state \(\uket\) belongs to one of the \emph{mutually orthogonal} non-null subspaces \(\subspace_1,\dots,\subspace_r\subseteq\hilbertspace\), with \(r\geq2\), but where You're ignorant about which subspace it is.
We'll let \(\subunion\coloneqq\bigcup_{k=1}^r\subspace_k\).
Due to the orthogonality of the non-null subspaces \(\subspace_1,\dots,\subspace_r\), \(\projection[\subspace_k]\projection[\subspace_\ell]=\projection[\subspace_\ell]\projection[\subspace_k]=\delta_{k,\ell}\projection[\subspace_k]\) for all \(k,\ell\in\set{1,\dots,r}\), but \(\subunion\) is no subspace of \(\hilbertspace\).

Once again, we start by looking at how Your set of desirable measurements \(\desirset\) can incorporate --- or reflect --- this knowledge, besides any other beliefs You might have.
So let's suppose that You know that \(\uket\in\subunion\).

First, consider any \(\measurement{A}\in\measurements\) such that
\begin{multline*}
\group{\forall k\in\set{1,\dots,r}}\projection[\subspace_k]\measurement{A}\projection[\subspace_k]=\zero,\\
\text{ or equivalently, }
\smash[t]{\sum_{k=1}^{r}\projection[\subspace_k]\measurement{A}\projection[\subspace_k]}=\zero.
\end{multline*}
Your knowledge that \(\uket\in\subunion\) tells You that there's a unique \(\ell\in\set{1,\dots,r}\) such that \(\uket\in\subspace_{\ell}\) and therefore \(\projection[\subspace_\ell]\uket=\uket\).
Hence, \(\ubra\measurement{A}\uket=\ubra\projection[\subspace_\ell]\measurement{A}\projection[\subspace_\ell]\uket=0\), so You're sure that the reward resulting from the measurement \(\measurement{A}\) will be zero: \(\measurement{A}\) is \emph{equivalent to the null measurement} \(\zero\), or in other words, indifferent, to You.

Conversely, if \(\measurement{A}\) is such that \(\sum_{k=1}^{r}\projection[\subspace_k]\measurement{A}\projection[\subspace_k]\neq\zero\), consider the \emph{non-empty} set \(K\) of those \(k\in\set{1,\dots,r}\) for which \(\projection[\subspace_k]\measurement{A}\projection[\subspace_k]\neq\zero\).
Your knowledge that \(\uket\in\subunion\) again tells You that there's a unique \(\ell\in\set{1,\dots,r}\) such that \(\uket\in\subspace_{\ell}\) and therefore \(\projection[\subspace_\ell]\uket=\uket\), but \emph{it in no way allows You to exclude that} \(\ell\in K\); You therefore can't exclude that \(\ubra\measurement{A}\uket=\ubra\projection[\subspace_\ell]\measurement{A}\projection[\subspace_\ell]\uket\neq0\).
In this case, You can't infer from the knowledge that \(\uket\in\subunion\) that the reward resulting from the measurement \(\measurement{A}\) will be zero.

This leads us to set up an analogous argumentation to the one in \cref{sec:conditioning}, and define the linear map \(\proj[\subunion]\colon\measurements\to\measurements\colon\measurement{A}\mapsto\sum_{k=1}^{r}\projection[\subspace_k]\measurement{A}\projection[\subspace_k]\).
Due to the mutual orthogonality of the subspaces \(\subspace_k\), we find that \(\proj[\subunion]\circ\proj[\subunion]=\proj[\subunion]\), so \(\proj[\subunion]\) is a linear projection operator on the real linear measurement space \(\measurements\), whose kernel is the set of indifferent measurements \(\indifset[\subunion]\coloneqq\set{\measurement{A}\in\measurements\given\proj[\subunion](\measurement{A})=\zero}=\bigcap_{k=1}^r\indifset[\subspace_k]\) corresponding to Your newly acquired knowledge that \(\uket\in\subunion\).

Here too, we can argue that a converse holds, that is, if You're indifferent to \(\indifset[\subunion]\), then You must believe that \(\uket\in\subunion\).
Suppose that You're indifferent to the measurements in \(\indifset[\subunion]\), and that You nevertheless believe that it's possible that \(\uket=\fket\) with \(\fket\notin\subunion\).
Then clearly \(\proj[\subunion](\fket\fbra)=\sum_{k=1}^{r}\projection[\subspace_k]\fket\fbra\projection[\subspace_k]\lneq\fket\fbra\).
Consider the measurement \(\measurement{A}\coloneq\proj[\subunion](\fket\fbra)-\fket\fbra\in\indifset[\subunion]\), then \(\utilitypure{\measurement{A}}(\fket)=-\fbra\proj[\subunion](\fket\fbra)\fket-\fbra\fket\fbra\fket<0\).
You're therefore indifferent to, and consequently have fair price zero for, a loss You deem possible, without any possibility of gain; this is unreasonable.

We emphasise that this argument is also completely analogous to a similar one for classical probabilities, where You also look at the intersection of the sets of indifferent gambles when You know one of their corresponding events occurs.

Since it holds that \(\measurement{A}\weakgeq\zero\then\projof[\subunion]{A}\weakgeq0\), all the arguments and results about the interplay of coherence, indifference and updating made in \Cref{sec:abstract} can be applied to this special case as well: suitably instantiated versions of \Cref{prop:indifference:subspace:general,prop:ordering:and:projection:general,prop:local:representation:general,prop:original:representation} and \cref{eq:conditioning:coherent:general:local,eq:conditioning:coherent:general:second} hold.
In particular, we find for Your coherent \emph{updated set of desirable measurements} in the reduced space \(\rng(\proj[\subunion])\) that, with obvious notations,
\begin{equation}\label{eq:conditioning:coherent:subunion:local:simpler}
\localcond[\subunion]{\desirset}
=\rngposdefmeasurements[\subunion]\cup\group[\big]{\desirset\cap\rng(\proj[\subunion])},
\end{equation}
and for the coherent \(\rng(\proj[\subunion])\)-focused set of desirable measurements \(\cond[\subunion]{\desirset}\) that corresponds to Your updated set \(\localcond[\subunion]{\desirset}\) of desirable measurements in \(\rng(\proj[\subunion])\) that
\begin{equation}\label{eq:conditioning:coherent:subunion:simpler}
\projcond{\desirset}
=\newbackground[\subunion]\cup\set{\measurement{A}\in\measurements\given\projof[\subunion]{A}\in\desirset}.
\end{equation}
Also, \cref{eq:prevision:update} readily transforms into \(\linprev\group{\measurement{A}\vert\subunion}=\linprev\group{\projof[\subunion]{A}}/\linprev\group{\projof[\subunion]{I}}\) for all~\(\measurement{A}\in\measurements\), leading, after a few manipulations, similar to the ones in \Cref{sec:back:to:conditioning}, to the following updating formula for density operators:
\begin{equation*}
\density[\subunion]
=\frac{\sum_{k=1}^r\projection[\subspace_k]\density\projection[\subspace_k]}{\trace{\sum_{k=1}^r\projection[\subspace_k]\density\projection[\subspace_k]}}
=\frac{\proj[\subunion]\group{\density}}{\trace{\proj[\subunion]\group{\density}}}.
\end{equation*}
This expression is a generalisation of Lüders conditioning, as it combines the Law of Total Probability with Lüders' rule as established in \Cref{sec:back:to:conditioning}.
More specifically, starting from a density operator \(\density\), the standard theory of quantum mechanics dictates that the probability of measuring a certain outcome \(\eigval_k\) --- an eigenvalue with corresponding eigenspace \(\subspace_k\) --- of a measurement \(\measurement{A}\) is given by \(\trace{\proj[{\subspace_k}]\group{\density}}\).
Lüders' rule tells us that after learning that the system's state belongs to the subspace \(\subspace_k\), the updated density operator is given by \(\density[\subspace_k]=\proj[{\subspace_k}]\group{\density}/\trace{\proj[{\subspace_k}]\group{\density}}\).
By the Law of Total Probability, then,
\begin{align*}
\density[\subunion]
&=\sum_{\ell=1}^r\density[\subspace_\ell]\frac{\trace{\proj[{\subspace_\ell}]\group{\density}}}{\sum_{k=1}^r\trace{\proj[{\subspace_k}]\group{\density}}}
=\frac{\sum_{\ell=1}^r\proj[{\subspace_k}]\group{\density}}{\sum_{k=1}^r\trace{\proj[{\subspace_k}]\group{\density}}}\\
&=\frac{\proj[\subunion]\group{\density}}{\trace{\proj[\subunion]\group{\density}}}.
\end{align*}
This shows that our more general theory of conditioning in quantum mechanics is compatible with the commonly used probabilistic framework, and that arguments of marginal extension \cite{Walley,miranda2007:marginal:extension} --- the law of total probability or iterated expectations --- can be used to derive our more general conditioning rule at least when the previsions involved are precise.
Whether such marginal extension can also be used in the case of imprecise (lower) previsions is a matter for further research.

\begin{runexample}
In our running example, You've instead learnt that the system's state \(\uket\) belongs to exactly one of the subspaces \(\subspace_1=\linspanof{\ket{-1,1}}\), \(\subspace_2=\linspanof{\ket{0,0}}\) and \(\subspace_3=\linspanof{\ket{1,-1}}\), for instance, by measuring the spins of both the particles but only registering their sum.
With \(\subunion\coloneqq\subspace_1\cup\subspace_2\cup\subspace_3\), the set of desirable measurements \(\cond[\subunion]{\desirset}\) is given by
\begin{equation*}
\cond[\subunion]{\desirset}
=\newbackground\cup\set{\measurement{A}\in\measurements\given\trace{\projof[\subunion]{A}\density}>0}.
\end{equation*}
The corresponding set of density operators is again a singleton \(\densities_{\cond{\desirset}}
=\set{\density[\subspace]}\), with
\begin{equation*}
\density[\subspace]
=\frac12\ket{-1,1}\bra{-1,1}+\frac12\ket{0,0}\bra{0,0}.
\end{equation*}
This density operator corresponds to the uniform probability over the states \(\ket{-1,1}\) and \(\ket{0,0}\).
This corresponds to the intuition that by measuring the state \(\gket=\nicefrac1{\sqrt3}(\ket{-1,1}-\ket{0,0}+\ket{1,1})\) along its basis, the results must be either \(\ket{-1,1}\), \(\ket{0,0}\) or \(\ket{1,1}\) with equal probability, but because You know that the total spin is \(0\), You end up with a uniform probability over \(\ket{-1,1}\) and \(\ket{0,0}\).
\end{runexample}

\section{Conclusion}
The desirable measurement framework has the advantage that it allows us at the same time to justify using Born's rule in quantum mechanics, and to extend it to situations where Your beliefs lead to partial preferences.
Furthermore, we've shown that it's possible to derive and generalise Lüder's rule for conditioning, by exploiting the interplay between desirability and indifference.
We stress that, in spirit, this is completely in line with ideas about updating classical sets of desirable gambles in the existing literature \cite{Walley,walley2000}.

In further work, we envisage extending this framework to deal with more general positive operator valued measures, by introducing marginalisation and cylindrical extension.
This work should be helpful in dealing with conservative inference in quantum computing and control.
In particular, it's a matter of future interest whether updating can be effected through marginal extension \cite{miranda2007:marginal:extension} also when uncertainty about the quantum state is expressed through imprecise probability models, as that would allow us to perform conservative inferences about successive measurements efficiently.

\appendix
\section*{Additional author information}

\begin{acknowledgements}
Funding for Gert de Cooman's research is partly covered by Ghent University's pioneering non-competitive research funding initiative.

We'd like to express our deep appreciation to four anonymous reviewers for taking our work seriously, for pointing out typos as well as small oversights, and for amongst them making quite a few suggestions for its improvement and possible continuation.
And where we didn't agree with them, their comments provided us with a clear motivation to clarify our arguments.
\end{acknowledgements}

\begin{authorcontributions}
This paper illustrates an interplay of ideas coming out of two different endeavours: Keano's exploring imprecision in quantum mechanics, and Gert's efforts to condense his ideas about the foundations of probability theory into a book: at some point it's no longer clear, or relevant, who influenced whom in coming up with what.
What is clear, is that Keano came up with the idea for the topic, and took the initiative in writing it all down, followed by a thorough revision and a reasonable amount of streamlining and rewriting by Gert.
Even at this later stage, active discussion between the two authors was crucial in shaping the final manuscript.
\end{authorcontributions}

\printbibliography

\end{document}